\documentclass[conference,final]{IEEEtran}
\usepackage[final]{graphicx}
\usepackage[reqno]{amsmath}
\usepackage{amssymb}
\newfont{\bbb}{msbm10 scaled 500}

\newfont{\bb}{msbm10 scaled 1100}

\newcommand{\Prob}{\textrm{Pr}}

\newcommand{\Ac}{{\cal A}}
\newcommand{\Bc}{{\cal B}}
\newcommand{\Cc}{{\cal C}}

\newcommand{\Fc}{{\cal F}}
\newcommand{\Gc}{{\cal G}}

\newcommand{\Mc}{{\cal M}}

\newcommand{\Xc}{{\cal X}}
\newcommand{\Yc}{{\cal Y}}

\newtheorem{theorem}{Theorem}
\newtheorem{corollary}[theorem]{Corollary}%[chapter]
\newtheorem{lemma}[theorem]{Lemma}%[chapter]
\newtheorem{definition}[theorem]{Definition}%[chapter]
\IEEEoverridecommandlockouts
\title{Polar Coding for Secure Transmission and Key Agreement}

\author{
\IEEEauthorblockN{O.~Ozan~Koyluoglu and Hesham El Gamal\\}
\IEEEauthorblockA{Department of Electrical and Computer Engineering\\
The Ohio State University\\
Columbus, OH 43210}
\thanks{This work is partially supported by Los Alamos National Labs (LANL)
and by National Science Foundation (NSF).
The first author is partially supported by the Presidential Fellowship
award of the Ohio State University.}
\thanks{This work will appear in Proc. 21st Annual IEEE
International Symposium on Personal, Indoor, and  Mobile Radio
Communications (PIMRC), Sept. 2010, Istanbul, Turkey.}
}

\begin{document}
\maketitle

%%%%%%%%%%%%%%%%%%%%%%%%%%%%%%%%%%%%%%%%%%%%%%%%%%%%%%%%%%%%%%%%%%%%%%%%%%%%%%
%%%%%%%%%%%%%%%%%%%%%%%%%%%%%%%%%%%%%%%%%%%%%%%%%%%%%%%%%%%%%%%%%%%%%%%%%%%%%%

\begin{abstract}
Wyner's work on wiretap channels and the recent works on information
theoretic security are based on random codes. Achieving information
theoretical security with practical coding schemes is of definite
interest. In this note, the attempt is to overcome this elusive task
by employing the polar coding technique of Ar{\i}kan. It is shown that
polar codes achieve non-trivial perfect secrecy rates for binary-input
degraded wiretap channels while enjoying their low encoding-decoding complexity.
In the special case of symmetric main and eavesdropper channels,
this coding technique achieves the secrecy capacity.
Next, fading erasure wiretap channels are considered and
a secret key agreement scheme is proposed, which requires only the
statistical knowledge of the eavesdropper channel state
information (CSI). The enabling factor is the creation of advantage
over Eve, by blindly using the proposed
scheme over each fading block, which is then exploited with
privacy amplification techniques to generate secret keys.
\end{abstract}

%%%%%%%%%%%%%%%%%%%%%%%%%%%%%%%%%%%%%%%%%%%%%%%%%%%%%%%%%%%%%%%%%%%%%%%%%%%%%%
%%%%%%%%%%%%%%%%%%%%%%%%%%%%%%%%%%%%%%%%%%%%%%%%%%%%%%%%%%%%%%%%%%%%%%%%%%%%%%

\section{Introduction}
\label{sec:Introduction}
The notion of information theoretic secrecy was introduced by
Shannon to study secure communication over point-to-point noiseless
channels~\cite{Shannon:Communication49}. This line of work
was later extended by Wyner~\cite{Wyner:The75} to noisy channels.
Wyner's degraded wiretap channel assumes that the eavesdropper
channel is a degraded version of the one seen by the legitimate
receiver. Under this assumption, Wyner showed that the
advantage of the main channel over that of the eavesdropper,
in terms of the lower noise level,
can be exploited to transmit secret bits using random codes.
This \emph{keyless secrecy} result was then extended to a more
general (broadcast) model in~\cite{Csiszar:Broadcast78} and to
the Gaussian setting in~\cite{Leung-Yan-Cheong:The78}. Recently,
there has been a renewed interest in wireless physical layer
security (see, e.g., Special Issue on Information
Theoretic Security, \emph{IEEE Trans. Inf. Theory}, June 2008
and references therein). However, designing practical codes
to achieve secrecy for any given main and eavesdropper channels
remained as an elusive task.

In~\cite{Thangaraj:Application07}, the authors constructed
LDPC based wiretap codes for certain binary erasure channel (BEC)
and binary symmetric channel (BSC) scenarios. In particular, when
the main channel is noiseless and the eavesdropper channel is
a BEC,~\cite{Thangaraj:Application07} presented codes that
approach secrecy capacity. For other scenarios, secrecy
capacity achieving code design is stated as an open problem.
Similarly,~\cite{Liu:Secure07} considers the design of secure
nested codes for the noiseless main channel setting (see
also~\cite{Liang:Information08}).

This work considers secret communication over a binary-input degraded
wiretap channel. Using the polar
coding technique of Ar{\i}kan~\cite{Arikan:Channel09},
we show that non-trivial secrecy rates are achievable.
According to our best knowledge, this coding technique is the
first provable and practical (having low encoding
and decoding complexity) secrecy encoding technique for this
set of channels.
In the special case of the symmetric main and eavesdropper channels,
this technique achieves the secrecy capacity of the channel
\footnote{We acknowledge that the concurrent work~\cite{Hessam:Achieving}
independently established the result that polar codes can achieve
the secrecy capacity of the degraded wiretap channels, when both
main and eavesdropper channels are binary-input and symmetric
(Corollary~\ref{thm:cor} of this note).}.
Next, we consider fading wiretap channels and propose a key
agreement scheme where the users only assumed to have the
statistical knowledge of the eavesdropper CSI.
The enabling observation is that by blindly using the
scheme over many fading blocks, the users will eventually create
an advantage over Eve, which can then be exploited to generate
secret keys using privacy amplification techniques.

%%%%%%%%%%%%%%%%%%%%%%%%%%%%%%%%%%%%%%%%%%%%%%%%%%%%%%%%%%%%%%%%%%%%%%%%%%%%%%
%%%%%%%%%%%%%%%%%%%%%%%%%%%%%%%%%%%%%%%%%%%%%%%%%%%%%%%%%%%%%%%%%%%%%%%%%%%%%%

\section{Notations}
\label{sec:Notations}

Throughout this paper, vectors are denoted by $x_1^N=\{x_1,\cdots,x_N\}$
or by $\bar{x}$ if we omit the indices.
Random variables are denoted with capital letters $X$, which
are defined over sets denoted by the calligraphic letters $\Xc$.
For a given set $\Ac\subset\{1,\cdots,N\}$, we write $x_{\Ac}$
to denote the sub-vector $\{x_i:i\in\Ac\}$. Omitting the
random variables, we use the following
shorthand for probability distributions
$p(x)\triangleq\Prob(X=x)$, $p(x|y)\triangleq\Prob(X=x|Y=y)$.

%%%%%%%%%%%%%%%%%%%%%%%%%%%%%%%%%%%%%%%%%%%%%%%%%%%%%%%%%%%%%%%%%%%%%%%%%%%%%%
%%%%%%%%%%%%%%%%%%%%%%%%%%%%%%%%%%%%%%%%%%%%%%%%%%%%%%%%%%%%%%%%%%%%%%%%%%%%%%

\section{Polar Codes}
\label{sec:PolarCodes}
Consider a binary-input DMC (B-DMC) given by $W(y|x)$, where $x\in\Xc=\{0,1\}$
and $y\in\Yc$ for some finite set $\Yc$. The $N$ uses of $W$ is
denoted by $W^N(y_1^N|x_1^N)$.
The symmetric capacity of a B-DMC $W$ is given by
\begin{eqnarray}
I(W) \triangleq \sum\limits_{x\in\Xc} \sum\limits_{y\in\Yc}
\frac{1}{2} W(y|x) \log_2\left(
\frac{W(y|x)}{\sum\limits_{x'\in\Xc} \frac{1}{2}W(y|x')}
\right),
\end{eqnarray}
which is the mutual information $I(X;Y)$ when the input $X$ is
uniformly distributed.
The Bhattacharyya parameter of $W$ is given by
\begin{eqnarray}
Z(W) \triangleq \sum\limits_{y\in\Yc} \sqrt{W(y|0)W(y|1)},
\end{eqnarray}
which measures the reliability of $W$ as it is an upper bound
on the probability of ML decision error on a single use
of the channel.

Polar codes is recently introduced by Ar{\i}kan~\cite{Arikan:Channel09}.
These codes can be encoded and decoded with complexity
$O(N\log N)$, while achieving an overall block-error
probability that is bounded as $O(2^{-N^\beta})$ for any
fixed $\beta<\frac{1}{2}$ (\cite{Arikan:Channel09},~\cite{Arikan:On09}).
In~\cite{Arikan:Channel09}, channel polarization is used to
construct codes (polar codes) that can achieve the symmetric
capacity, $I(W)$, of any given B-DMC $W$.
Channel polarization consists of two operations:
Channel combining and channel splitting. Let $u_1^N$ be the
vector to be transmitted. The combined channel is represented
by $W_N$ and is given by
\begin{eqnarray}
W_N(y_1^N|u_1^N) = W^N (y_1^N|u_1^N B_N F^{\otimes n}),
\end{eqnarray}
where $B_N$ is a bit-reversal permutation matrix, $N=2^n$, and
$F\triangleq \left(
               \begin{array}{cc}
                 1 & 0 \\
                 1 & 1 \\
               \end{array}
             \right)
$. Note that the actual channel input here is
given by $x_1^N=u_1^N B_N F^{\otimes n}$.
The channel splitting constructs $N$ binary input channels
from $W_N$, where the transformation is given by
\begin{eqnarray}
W_N^{(i)} (y_1^N,u_1^{i-1}|u_i) \triangleq
\sum\limits_{u_{i+1}^N \in \Xc^{N-i}}
\frac{1}{2^{N-1}} W_N(y_1^N|u_1^N).
\end{eqnarray}

The polarization phenomenon is shown by the following theorem.
\begin{theorem}[Theorem 1 of~\cite{Arikan:Channel09}]
For any B-DMC $W$, $N=2^n$ for some $n$, and $\delta\in(0,1)$,
we have
$$\lim\limits_{N\to\infty} \frac{|\{i\in\{1,\cdots,N\}:
I(W_N^{(i)})\in(1-\delta,1]\}|
}
{N}= I(W),$$
$$\lim\limits_{N\to\infty} \frac{|\{i\in\{1,\cdots,N\}:
I(W_N^{(i)})\in[0,\delta)\}|
}
{N}= 1-I(W).$$
\end{theorem}

In order to derive the rate of the channel polarization,
the random process $Z_n$ is defined in~\cite{Arikan:Channel09}
and in~\cite{Arikan:On09}. Basically,
\begin{eqnarray}
\Prob\{Z_n\in(a,b)\} = \frac{
|\{i\in\{1,\cdots,N\}:Z(W_{2^n}^{(i)})\in(a,b)\}|
}{N}
\end{eqnarray}
The rate of the channel polarization is given
by the following.
\begin{theorem}[Theorem 1 of~\cite{Arikan:On09}]
For any B-DMC $W$ and for any given $\beta<\frac{1}{2}$,
$$\lim\limits_{n\to\infty} \Prob\{Z_n < 2^{-{2^n}^{\beta}} \} = I(W).$$
\end{theorem}

Now, the idea of polar coding is clear. The encoder-decoder
pair, utilizing the polarization effect, will transmit
data through the subchannels for which $Z(W_{N}^{(i)})$
is near $0$.
In~\cite{Arikan:Channel09}, the polar code
$(N,K,\Ac,u_{\Ac^c})$ for B-DMC $W$ is defined by
$x_1^N=u_1^N B_N F^{\otimes n}$, where $u_{\Ac^c}$
is a given frozen vector, and the information set $\Ac$
is chosen such that $|\Ac|=K$ and
$Z(W_{N}^{(i)})<Z(W_{N}^{(j)})$ for all $i\in \Ac$,
$j\in \Ac^c$. The frozen vector $u_{\Ac^c}$ is given
to the decoder. Ar{\i}kan's successive cancellation (SC)
estimates the input as follows:
For the frozen indices $\hat{u}_{\Ac^c}=u_{\Ac^c}$.
For the remaining indices s.t. $i\in\Ac$;
$\hat{u}_i=0$, if
$W_N^{(i)} (y_1^N,\hat{u}_1^{i-1}|0)
\geq W_N^{(i)} (y_1^N,\hat{u}_1^{i-1}|1)$ and
$\hat{u}_i=1$, otherwise.
With this decoder, it is shown in~\cite{Arikan:Channel09}
that the average block
error probability over the ensemble
(consisting of all possible frozen vector choices)
of polar codes is bounded by
$$P_e(N)\leq \sum\limits_{i\in\Ac} Z(W_{N}^{(i)}).$$

We now state the result of~\cite{Arikan:Channel09}
using the bound given in~\cite{Arikan:On09}.

\begin{theorem}[Theorem 2 of~\cite{Arikan:On09}]
For any given B-DMC $W$ with $I(W)>0$, let $R<I(W)$
and $\beta\in(0,\frac{1}{2})$ be fixed. Block error probability
for polar coding under SC decoding (averaged over
possible choices of frozen vectors) satisfies
$$P_e(N)= O(2^{-N^\beta}).$$
\end{theorem}

Note that, for any given $\beta\in(0,\frac{1}{2})$ and
$\epsilon>0$, we can define the sequence of polar codes
by choosing the information indices as
$$\Ac_N=\{i\in\{1,\cdots,N\}: Z(W_{N}^{(i)}) \leq
\frac{1}{N}2^{-N^\beta}\}.$$
Then, from the above theorems, for sufficiently large $N$,
we can achieve the rate
$$R=\frac{|\Ac_N|}{N}\geq I(W)-\epsilon$$
with average block error probability (averaged over the
possible choices of $u_{\Ac_N^c}$)
$$P_e(N)\leq \sum\limits_{i\in\Ac_N} Z(W_{N}^{(i)}) \leq
2^{-N^\beta}$$
under SC decoding. (See also~\cite{Korada:Polar09}.)

This result shows the existence of a polar code $(N,K,\Ac,u_{\Ac^c})$
achieving the symmetric capacity of $W$. We remark that, any
frozen vector choice of $u_{\Ac^c}$ will work for
symmetric channels~\cite{Arikan:Channel09}.
For our purposes, we will denote a polar code for B-DMC $W$ with
$\Cc(N,\Fc,u_{\Fc})$, where the frozen set is given by
$\Fc \triangleq \Ac^c$.
Note that, $\Ac$ denotes the indices of information transmission
for the polar code, whereas $\Fc$ is the set of frozen indices.

We conclude this section by noting the following lemma
(given in~\cite{Korada:Polar09}) regarding polar coding
over degraded channels.

\begin{lemma}[Lemma 4.7 of~\cite{Korada:Polar09}]\label{thm:lemma}
Let $W:\Xc\to\Yc$ and $W':\Xc\to\Yc'$ be two B-DMCs such that
$W$ is degraded w.r.t. $W'$, i.e., there exists a channel
$W'':\Yc'\to\Yc$ such that
$$W(y|x)=\sum\limits_{y'\in\Yc'}
W'(y'|x)W''(y|y').$$
Then, $W_N^{(i)}$ is degraded w.r.t. ${W'}_N^{(i)}$
and $Z(W_N^{(i)})\geq Z({W'}_N^{(i)})$.
\end{lemma}

%%%%%%%%%%%%%%%%%%%%%%%%%%%%%%%%%%%%%%%%%%%%%%%%%%%%%%%%%%%%%%%%%%%%%%%%%%%%%%
%%%%%%%%%%%%%%%%%%%%%%%%%%%%%%%%%%%%%%%%%%%%%%%%%%%%%%%%%%%%%%%%%%%%%%%%%%%%%%

\section{Secure Transmission over Wiretap Channel}
\label{sec:WiretapChannel}

A discrete memoryless wiretap channel with
is denoted by
$$(\Xc, W(y_m,y_e|x),
\Yc_m \times \Yc_e),$$ for some finite sets
$\Xc, \Yc_m, \Yc_e$.
Here the symbols $x\in \Xc$ are the
channel inputs and the symbols $(y_m,y_e)\in \Yc_m
\times \Yc_e$ are the channel outputs observed at the
main decoder and at the eavesdropper, respectively.
The channel is memoryless and
time-invariant:
$$p({y_m}_i,{y_e}_i|{x}_1^i,
{y_m}_1^{i-1},{y_e}_1^{i-1})= W({y_m}_i,{y_e}_i|{x}_i).$$
We assume that the transmitter has a secret message
$M$ which is to be transmitted to the receiver in $N$
channel uses and to be secured from the eavesdropper.
In this setting, a secret codebook has the following components:

$1$) The secret message set $\Mc$.
The transmitted messages are assumed to be uniformly
distributed over these message sets.

$2$) A stochastic encoding function $f(.)$ at the transmitter
which maps the secret messages to the transmitted symbols:
$f:m\to {X}_1^N$ for each $m\in\Mc$.

$3$) Decoding function $\phi(.)$ at receiver
which maps the received symbols to estimate of the message:
$\phi({Y_m}_1^N)=\{\hat{m}\}$.

The reliability of transmission is measured by the following
probability of error.
\begin{IEEEeqnarray}{l}
P_{e}=\frac{1}{|\Mc|}\sum\limits_{(m)\in \Mc}
\textrm{Pr} \left\{\phi({Y_m}_1^N)\neq (m)| (m)
\textrm{ is sent}\right\}\nonumber
\end{IEEEeqnarray}
We say that the rate $R$ is an achievable secrecy rate,
if, for any given $\epsilon > 0$, there exists a
secret codebook such that,
\begin{eqnarray}\label{eq:secrecy}
\frac{1}{N}\log(|\Mc|) &=& R \nonumber\\
P_{e}&\leq& \epsilon \nonumber\\
\frac{1}{N}I\left(M;{Y_e}_1^N\right) &\leq& \epsilon
\end{eqnarray}
for sufficiently large $N$.

Consider a degraded binary-input wiretap channel,
where, for the input set $\Xc=\{0,1\}$, the
main channel is given by
\begin{eqnarray}
W_m(y_m|x)
\end{eqnarray}
and the eavesdropper channel is
\begin{eqnarray}
W_e(y_e|x)=\sum\limits_{y_m\in\Yc_m}
W_m(y_m|x)W_d(y_e|y_m).
\end{eqnarray}
Here, the degradation is due to the channel $W_d(y_e|y_m)$.

Note that, due to degradation, polar codes designed for the
eavesdropper channel can be used for the main channel.
For a given sufficiently large $N$ and $\beta\in(0,\frac{1}{2})$,
let
$$\Ac_m=\{i\in\{1,\cdots,N\}: Z({W_m}_{N}^{(i)}) \leq
\frac{1}{N}2^{-N^\beta}\},$$
$$\Ac_e=\{i\in\{1,\cdots,N\}: Z({W_e}_{N}^{(i)}) \leq
\frac{1}{N}2^{-N^\beta}\}.$$
Now, consider a polar code $\Cc_m\triangleq\Cc(N,\Fc_m,u_{\Fc_m})$
for the main channel with some $u_{\Fc_m}$.
Due to Lemma~\ref{thm:lemma}, we have
$\Ac_e \subset \Ac_m$ and hence $\Fc_m \subset \Fc_e$.
Now, for any given length $|\Fc_e|-|\Fc_m|$ vector
$\bar{v}_m$ and $u_{\Fc_m}$,
we define the frozen vector for the eavesdropper,
denoted by $u_{\Fc_e}(\bar{v}_m)$, by choosing
$(u_{\Fc_e}(\bar{v}_m))_{\Fc_m}=u_{\Fc_m}$ and
$(u_{\Fc_e}(\bar{v}_m))_{\Fc_e \backslash \Fc_m}=\bar{v}_m$.
Note that, denoting $\Cc_e(\bar{v}_m)\triangleq
\Cc(N,\Fc_e,u_{\Fc_e}(\bar{v}_m))$, the ensemble
$\cup_{\bar{v}_m,u_{\Fc_m}} \Cc_e(\bar{v}_m)$
is a symmetric capacity achieving polar code ensemble
for the eavesdropper channel $W_e$
(if the eavesdropper channel is symmetric, any frozen vector
choice will work~\cite{Arikan:Channel09}, and hence
the code achieves the capacity of the eavesdropper channel for
any $\bar{v}_m,u_{\Fc_m}$).
This implies that the code for the
main channel can be partitioned as
$\Cc_m=\cup_{\bar{v}_m} \Cc_e(\bar{v}_m)$.
This observation, when considered over the ensemble of codes,
enables us to construct secrecy achieving polar coding schemes,
even if the eavesdropper channel is not symmetric, as
characterized by the following theorem.

\begin{theorem}\label{thm:Wiretap}
For a binary-input degraded wiretap channel,
the perfect secrecy rate of $I(W_m)-I(W_e)$
is achieved by polar coding.
\end{theorem}

\begin{proof}

\textbf{Encoding:}
We map the secret message to be transmitted to
$\bar{v}_m$ and generate a random vector $\bar{v}_r$,
according to uniform distribution over $\Xc$,
of length $|\Ac_e|$. Then, the channel input
is constructed with $x_1^N=u_1^N B_N F^{\otimes n}$,
where $u_{\Fc_m}$ is the frozen vector of the
polar code $\Cc_m$, $u_{\Fc_e \backslash \Fc_m}=\bar{v}_m$,
and $u_{\Ac_e}=\bar{v}_r$. The polar code ensemble is
constructed over all different choices of frozen vectors,
i.e., $u_{\Fc_m}$.

\textbf{Decoding:}
The vectors $\bar{v}_m$ and $\bar{v}_r$ can be decoded
with the SC decoder described above with error probability
$P_e = O(2^{-N^{\beta}})$ (averaged over the ensemble)
achieving a rate $R=\frac{|\bar{v}_m|}{N}=I(W_m)-I(W_e)$
for sufficiently large $N$.

\textbf{Security:}
Lets assume that the vector $\bar{v}_m$ is given
to the eavesdropper along with $u_{\Fc_m}$. Then,
employing the SC decoding, the eavesdropper can decode
the random vector $\bar{v}_r$ with
$P_e = O(2^{-N^{\beta}})$ averaged over the ensemble.
Utilizing the Fano's
inequality and average it over the code ensemble
seen by the Eve, i.e. over $\bar{V}_m$ and $U_{\Fc_m}$,
we obtain
\begin{eqnarray}\label{eq:Fanos}
H(\bar{V}_r|\bar{V}_m,U_{\Fc_m},{Y_e}_1^N)
\leq H(P_e) + N \log(|\Xc|)P_e
\leq N \epsilon(N),
\end{eqnarray}
where $\epsilon(N)\to 0$ as $N\to \infty$.

Then, the mutual information leakage to the eavesdropper
averaged over the ensemble can be bounded as follows.

$I(M;{Y_e}_1^N|U_{\Fc_m}) =
I(\bar{V}_m;{Y_e}_1^N|U_{\Fc_m})$
\begin{eqnarray}
&=& I(\bar{V}_m,\bar{V}_r;{Y_e}_1^N|U_{\Fc_m})-
I(\bar{V}_r;{Y_e}_1^N|\bar{V}_m,U_{\Fc_m})\\
&\stackrel{(a)}{=}& I(U_1^N;{Y_e}_1^N) - H(\bar{V}_r) + H(\bar{V}_r|\bar{V}_m,U_{\Fc_m},{Y_e}_1^N) \\
&\stackrel{(b)}{\leq}& I(X_1^N;{Y_e}_1^N) - H(\bar{V}_r) + H(\bar{V}_r|\bar{V}_m,U_{\Fc_m},{Y_e}_1^N) \\
&\stackrel{(c)}{\leq}& NI(W_e) - |\Ac_e| + H(\bar{V}_r|\bar{V}_m,u_{\Fc_m},{Y_e}_1^N) \\
&\stackrel{(d)}{\leq}& N I(W_e) - |\Ac_e| + N \epsilon(N),
\end{eqnarray}
where in (a) we have $U_1^N$ each entry with i.i.d. uniformly
distributed, (b) follows from data processing inequality,
(c) is due to $I(X_1^N;{Y_e}_1^N)=
\sum\limits_{i=1}^N I(X_1^N;{Y_e}_i|{Y_e}_1^{i-1})\leq
\sum\limits_{i=1}^N H({Y_e}_i) - H({Y_e}_i|{X}_i)=
N I(X_i;{Y_e}_i)$
with a uniformly distributed $X_i$, and (d) follows from \eqref{eq:Fanos}
with $\epsilon(N)\to 0$ as $N\to \infty$.
As $\frac{|\Ac_e|}{N}\to I(W_e)$ as $N$ gets large,
we obtain
\begin{eqnarray}
\frac{1}{N}I(\bar{V}_m;{Y_e}_1^N|U_{\Fc_m}) &\leq& \epsilon
\end{eqnarray}
for a given $\epsilon>0$ for sufficiently large $N$.
As the reliability and secrecy constraints are satisfied averaged
over the ensemble, there exist a polar code with some fixed $u_{\Fc_m}$
achieving the secure rate $I(W_m)-I(W_e)$.
\end{proof}

Note that in the above result, the code satisfying the reliability and
the secrecy constraints can be found from the ensemble by an exhaustive
search. However, as block length increases, almost all the codes in the
ensemble will do equally well.
If the eavesdropper channel is symmetric, then the secrecy constraint is
satisfied for any given frozen vector $u_{\Fc_m}$ and the code search
is only for the reliability constraint. If the eavesdropper channel
is not symmetric, a prefix channel can be utilized to have this
property.
\begin{corollary}
For non-symmetric eavesdropper channels, the channel can be
prefixed with some $p(x|x')$ such that the resulting
eavesdropper channel
$$W_e'(y_e|x')=\sum\limits_{y_m\in\Yc_m}
p(x|x')W_m(y_m|x)W_d(y_e|y_m)$$
is symmetric. Then, using the scheme above, the secret rate
$$R=I(W_m')-I(W_e')$$
is achievable, where $W_m'(y_m|x')=p(x|x')W_m(y_m|x)$.
\end{corollary}
Finally, we note that the scheme achieves the secrecy capacity
and any code in the ensemble, i.e., any fixed $u_{\Fc_m}$,
will satisfy both the reliability and secrecy constraints, if the
main and eavesdropper channels are symmetric.
\begin{corollary}\label{thm:cor}
For a binary-input degraded wiretap channel with
symmetric main and eavesdropper channels,
polar coding achieves the secrecy capacity,
i.e., $C(W_m)-C(W_e)$, of the channel.
\end{corollary}

We note that the stated results are achievable by encoders and
decoders with complexity of $O(N \log N)$ for each.
In addition, if the channels are binary erasure
channels (BECs), then there exists algorithms with
complexity $O(N)$ for the code construction~\cite{Arikan:Channel09}.

%%%%%%%%%%%%%%%%%%%%%%%%%%%%%%%%%%%%%%%%%%%%%%%%%%%%%%%%%%%%%%%%%%%%%%%%%%%%%%
%%%%%%%%%%%%%%%%%%%%%%%%%%%%%%%%%%%%%%%%%%%%%%%%%%%%%%%%%%%%%%%%%%%%%%%%%%%%%%

\section{Secret Key Agreement over Fading Wiretap Channels}
\label{sec:Fading}

In this section, we focus on the following key agreement problem:
Alice, over fading wiretap channel, would like to
agree on a secret key with Bob in the presence of passive eavesdropper
Eve. We focus on the special case of binary erasure main and
eavesdropper channels, for which the code construction is
shown to be simple~\cite{Arikan:Channel09}.

Fading blocks are represented by $i=1,\cdots,LM$ and
each block has $N$ channel uses.
Random variables over blocks are represented with the following
bar notation. $\bar{Y}_e^{(l;m)}$ denotes the observations of Eve
over the fading block $m$ of the super block $l$,
the observations of Eve over super block $l\in[1,L]$ is
denoted by $\bar{\bar{Y}}_e^{(l)}=
\bar{Y}_e^{(l;1\cdots M)} \triangleq
\{\bar{Y}_e^{(l;1)},\cdots,\bar{Y}_e^{(l;M)} \}$, and
Eve's total observation over all
super blocks is denoted by $Y_e^*=\bar{\bar{Y}}_e^{(1\cdots L)}
=\{\bar{\bar{Y}}_e^{(1)},\cdots,\bar{\bar{Y}}_e^{(L)}\}$.

Main and eavesdropper channels are binary erasure channels
and are denoted by $W_m^{(i)}$
and $W_e^{(i)}$, respectively.
Here, the channels $W_m$ and $W_e$ are random, outcome of which
result in the channels of each block.
Instantaneous eavesdropper CSI is not known at the users,
only the statistical knowledge of it is assumed.
The channels are assumed to be physically degraded w.r.t.
\emph{some} order at each block.
\footnote{Remarkable, a random walk model with packet erasures can
be covered with this model. Also, parallel channel model is
equivalent to this scenario.}
Note that, in this setup, eavesdropper channel
can be better than the main channel on the average.

We utilize the proposed secrecy encoding scheme for the wiretap
channel at each fading block. Omitting the block indices,
frozen and information bits are denoted as $u_{\Fc_m}$ and
$u_{\Ac_m}$, respectively.
Information bits are uniformly distributed
binary random variables and are mapped to $u_{\Ac_m}$.
Secret and randomization bits among these information bits
are denoted by $\bar{V}_m$ and $\bar{V}_r$, respectively.
Frozen bits are provided both to main receiver and eavesdropper
at each block. (We omitted writing this side information below
as all zero vector can be chosen as the frozen vector for
the erasure channel~\cite{Arikan:Channel09}.)
Note that Alice and Bob do not know the length of $\bar{V}_m^{(i)}$
at fading block $i$. In particular, there may not be any secured bits
at a given fading block.

Considering the resulting information accumulation over a block,
we obtain the followings.
\begin{eqnarray}
\frac{1}{N}H(\bar{V}_m^{(i)}) &=& [C(W_m^{(i)}) - C(W_e^{(i)})]^+\nonumber\\
\frac{1}{N}H(\bar{V}_r^{(i)}) &=& \min\{C(W_m^{(i)}),C(W_e^{(i)})\},\nonumber
\end{eqnarray}
where the former denotes the amount of secure information generated at
block $i$ (here the secrecy level is the bound on the mutual information
leakage rate), and the latter denotes the remaining information.
Note that these entropies are random variables as channels are random
over the blocks. Remarkable, this scheme converts the fading
phenomenon to the advantage of Alice and Bob
(similar to the enabling observation utilized in~\cite{Gopala:On08}).
Exploiting this observation and coding over $LM$ fading blocks, the
proposed scheme below creates advantage for the main users:
As $L,M,N$ get large, information bits, denoted by $W^*$, are
w.h.p. reliably decoded at the Bob,
$H(W^*) \to L M N \: E\left[ C(W_m) \right]$,
and
$H(W^*|Y_e^*) \to L M N \: E\left[ [C(W_m) - C(W_e)]^+ \right]$.
This accomplishes \emph{both} advantage distillation and information
reconciliation phases of a key agreement
protocol~\cite{Bennett:Generalized95,Bloch:WirelessI}.
Now, a third phase (called as \emph{privacy amplification})
is needed to distill a shorter string $K$ from $W^*$, about which
Eve has only a negligible amount of information. The privacy
amplification step can be done with universal hashing
as considered in~\cite{Bennett:Generalized95}.
We first state the following definitions and lemma regarding
universal hashing, and then formalize the main result of this section
in the following theorem.

\begin{definition}
A class $\Gc$ of functions $\Ac \to \Bc$ is universal if, for
any $x_1 \neq x_2$ in $\Ac$, the probability that $g(x_1)=g(x_2)$
is at most $\frac{1}{|\Bc|}$ when $g$ is chosen as random from $\Gc$
according to the uniform distribution.
\end{definition}

There are efficient universal classes, e.g., to map $n$ bits to
$r$ bits, class of linear functions given by $r \times n$
matrices needs $rn$ bits to describe~\cite{Carter:Universal79}.
Note that hash function should have complexity as 1) it will be
revealed to each user, and 2) Alice and Bob will compute $g(W^*)$.
There are more efficient classes with polynomial time evaluation complexity
and $O(n)$ description complexity~\cite{Carter:Universal79}.

Generalized privacy amplification, proposed in~\cite{Bennett:Generalized95},
is based on the following property of universal hashing.
\begin{lemma}[Theorem 3,~\cite{Bennett:Generalized95}]\label{thm:PA}
Let $X\in\Xc$ be a random variable with distribution $P_X$ and
R{\'e}nyi entropy (of second order) $R(X)= -\log_2 E[P_X(X)]$.
Let $G$ be a random choice (according to
uniform distribution) of a member of universal class of hash functions
$\Xc \to \{0,1\}^r$, and let $Q=G(X)$. Then, we have
$$H(Q|G)\geq R(Q|G) \geq r- \log_2 \left( 1 + 2^{r-R(X)} \right)
\geq r- \frac{2^{r-R(X)}}{\ln 2}.$$
\end{lemma}

Exploiting the proposed coding scheme, which creates advantage in favor
of Bob over the fading channel, we use the hash functions described
above and obtain the following result.
\begin{theorem}
For any $\epsilon,\epsilon^*>0$, let
$$n=L \: M \: N \: \left(E\left[ C(W_m) \right]-\epsilon^*\right),$$
$$r= L \: M \: N \: \left( E\left[ [C(W_m) - C(W_e)]^+ \right] - \epsilon^* \right).$$
Then, for sufficiently large $L$, $M$ and $N$,
Alice and Bob can w.h.p. agree on
the random variable $W^*\triangleq \bar{\bar{W}}^{(1\cdots L)}$
of length $n$ over $LM$ fading blocks (i.e.,
$\Prob\{W^*\neq \hat{W}^*\} \leq \epsilon$, where
$\hat{W^*}$ denotes the estimate at Bob); and
choose $K=G(W^*)$ as their secret key (here $G$ is
chosen uniformly random from universal class of hash functions
$\{0,1\}^n \to \{0,1\}^r$) satisfying
$$I(K;Y_e^*, G) \leq \epsilon,$$
where $Y_e^* \triangleq \bar{\bar{Y}}_e^{(1\cdots L)}$
denotes the Eve's total received symbols.
\end{theorem}

\begin{proof}

We repeat the described scheme over $LM$ fading blocks.
Due to the construction above, we have
\begin{equation}\label{eq:2}
\frac{1}{N} H(\bar{V}_m^{(i)}) - \epsilon_1
\leq \frac{1}{N} H(\bar{V}_m^{(i)}|\bar{Y}_e^{(i)})
\leq \frac{1}{N} H(\bar{V}_m^{(i)}),
\end{equation}
where $\frac{1}{N} H(\bar{V}_m^{(i)}) = [C(W_m^{(i)}) - C(W_e^{(i)})]^+$
and $\epsilon_1\to 0$ as $N$ gets large
(follows from the fact that conditioning does not increase entropy and the
security of $\bar{V}_m^{(i)}$),
and
\begin{equation}\label{eq:3}
\frac{1}{N} H(\bar{V}_r^{(i)}|\bar{Y}_e^{(i)}, \bar{V}_m^{(i)}) \leq \epsilon_2,
\end{equation}
where $\epsilon_2\to0$ as $N\to\infty$ (follows from Fano's inequality).

We now consider the total information accumulation and leakage.
Let $W^*=\bar{\bar{W}}^{(1\cdots L)}\triangleq \{\bar{V}_m^{(l;m)},\bar{V}_r^{(l;m)},
\forall l\in[1,L], \forall m\in[1,M]\}$ and denote the estimate
of it at Bob as $\hat{W}^*$.
We obtain that, there exist
$N_1,M_1$, s.t. for any $N\geq N_1$ and $M\geq M_1$, we have
\begin{equation}\label{eq:fading1}
H(W^*) \geq L M N \left(E\left[ C(W_m) \right]-\epsilon^*\right)
\end{equation}
\begin{equation}\label{eq:fading2}
\Prob\{W^*\neq \hat{W}^*\} \leq LM 2^{-N^{\beta}},
\end{equation}
for some $\beta\in (0,\frac{1}{2})$ due to polar coding and
the union bound.

Considering $Y_e^*\triangleq \bar{\bar{Y}}_e^{(1\cdots L)}$
at Eve, we write

$H(W^*|Y_e^*)= \sum\limits_{l=1}^{L}
H(\bar{\bar{W}}^{(l)}|\bar{\bar{Y}}_e^{(l)})$
\begin{eqnarray}\label{eq:1}
=\sum\limits_{i=1}^{LM} H( \bar{V}_m^{(i)} | \bar{Y}_e^{(i)})
+ H( \bar{V}_r^{(i)} | \bar{Y}_e^{(i)}, \bar{V}_m^{(i)}).
\end{eqnarray}
Focusing on a particular super block, omitting the index $(l)$ in
$(\bar{\bar{W}}^{(l)},\bar{\bar{Y}}_e^{(l)})$, and using
\eqref{eq:2} and \eqref{eq:3} in \eqref{eq:1}, we obtain

$MN \left( E\left[ [C(W_m) - C(W_e)]^+ \right] - \epsilon_4 \right)
\leq H(\bar{\bar{W}}|\bar{\bar{Y}}_e)$
\begin{eqnarray}
\leq MN \left( E\left[ [C(W_m) - C(W_e)]^+ \right] + \epsilon_5 \right),
\end{eqnarray}
where $\epsilon_4$ and $\epsilon_5$ vanishes as $M,N$ get large.

In order to translate $H(W^*|Y_e^*)$ to R{\'e}nyi entropy,
to use Lemma~\ref{thm:PA} in our problem, we resort
to typical sequences, as for a uniform random variable both
measures are the same.
Considering $(\bar{\bar{W}}^{(1)},\cdots,\bar{\bar{W}}^{(L)},
\bar{\bar{Y}}_e^{(1)},\cdots,\bar{\bar{Y}}_e^{(L)})$ as $L$
repetitions of the experiment of super block random variables
$(\bar{\bar{W}},\bar{\bar{Y}}_e)$, we define the event $T$
based on typical sets as follows~\cite{Maurer:Information00}:
Let $\delta>0$. $T=1$, if the sequences $\bar{\bar{w}}^{(1\cdots L)}$
and $(\bar{\bar{w}}^{(1\cdots L)},\bar{\bar{y}}_e^{(1\cdots L)})$
are $\delta$-typical; and $\bar{\bar{y}}_e^{(1\cdots L)}$ is
such that the probability that $({\bar{\bar{w}}'}^{(1\cdots L)},\bar{\bar{y}}_e^{(1\cdots L)})$
is $\delta$-typical is at least $1-\delta$, which is taken
over ${\bar{\bar{w}}'}^{(1\cdots L)}$ according to
$p({\bar{\bar{W}}'}^{(1\cdots L)}|\bar{\bar{y}}_e^{(1\cdots L)})$.
Otherwise, we set $T=0$ and denote $\delta_0\triangleq \Prob\{T=0\}$.
Then, by Lemma 6 of~\cite{Maurer:Information00}, as $L\to\infty$
\begin{equation}\label{eq:PA1}
L \delta_0 \to 0, L \delta\to 0,  \textrm{ and }
\end{equation}
\begin{eqnarray}\label{eq:PA2}
R(\bar{\bar{W}}^{(1\cdots L)}| \bar{\bar{Y}}_e^{(1\cdots L)}=\bar{\bar{y}}_e^{(1\cdots L)},T=1)\nonumber\\
\geq L (H(\bar{\bar{W}} | \bar{\bar{Y}}_e) - 2\delta)+\log(1-\delta).
\end{eqnarray}
We continue as follows.

$R(\bar{\bar{W}}^{(1\cdots L)}| \bar{\bar{Y}}_e^{(1\cdots L)}=\bar{\bar{y}}_e^{(1\cdots L)},T=1)$
\begin{eqnarray}\label{eq:PA3}
&\geq& L (H(\bar{\bar{W}} | \bar{\bar{Y}}_e) - 2\delta)+\log(1-\delta)\nonumber\\
&\geq& L M N \bigg( E\left[ [C(W_m) - C(W_e)]^+ \right]
- \epsilon_4 \nonumber\\
&&{-}\: \frac{2\delta}{MN} + \frac{\log(1-\delta)}{LMN}\bigg)\nonumber\\
&=& L M N \left( E\left[ [C(W_m) - C(W_e)]^+ \right]
- \delta^* \right),
\end{eqnarray}
where $\delta^* \to 0$ as $M,N\to\infty$.
Thus, for the given $\epsilon^*$, there exists $M_2,N_2$ s.t.
for $M\geq M_2$ and $N\geq N_2$, $\frac{\epsilon^*}{2} \geq \delta^*$. We let
$r=L M N \left( E\left[ [C(W_m) - C(W_e)]^+ \right]
- \epsilon^* \right)$ and consider the following bound.

$H(K|Y_e^*,G) \geq H(K|Y_e^*,G,T)$
\begin{eqnarray}
&\stackrel{(a)}{\geq}& (1-\delta_0) \sum\limits_{y_e^*\in \Yc_e^*}
\bigg(H(K|Y_e^*=y_e^*,G,T=1)\nonumber\\
&&P(Y_e^*=y_e^*|T=1) \bigg)\nonumber
\end{eqnarray}
\begin{eqnarray}
&\stackrel{(b)}{\geq}& (1-\delta_0) \left(r-\frac{2^{-LMN(\epsilon^*-\delta^*)}}{\ln 2}\right), \label{eq:fading3}
\end{eqnarray}
where in (a) $\delta_0$ is s.t. $L \delta_0 \to 0$ as $L\to\infty$, (b)
is due to Lemma~\ref{thm:PA} given above and due to
\eqref{eq:PA3} and the choice of $r$.
Here, for the given $\epsilon>0$, there exists $M_3,N_3$ s.t.
for $M\geq M_3$ and $N\geq N_3$,
$\frac{2^{-LMN(\frac{\epsilon^*}{2})}}{\ln 2}\leq \frac{\epsilon}{2}$.
Hence, we obtain
\begin{eqnarray}
I(K;Y_e^*,G) &=& H(K)-H(K|Y_e^*,G) \\
&\leq& \delta_0 r + \frac{2^{-LMN(\epsilon^*-\delta^*)}}{\ln 2} \\
&\stackrel{(a)}{\leq}& \delta_0 LMN + \frac{2^{-LMN(\frac{\epsilon^*}{2})}}{\ln 2}\\
&\stackrel{(b)}{\leq}& \delta_0 LMN + \frac{\epsilon}{2},\label{eq:fading4}
\end{eqnarray}
where (a) holds if $M\geq M_2$ and $N\geq N_2$ and (b) holds if
$M\geq M_3$ and $N\geq N_3$.

Now, we choose some $M\geq \max\{M_1,M_2,M_3\}$. For this choice of $M$,
we choose sufficiently large $L$ and sufficiently large $N$ such that
$N\geq \max\{N_1,N_2,N_3\}$ and
\begin{eqnarray}
\delta_0 LMN &\leq& \frac{\epsilon}{2} \label{eq:fading5}\\
LM2^{-N^{\beta}} &\leq& \epsilon \label{eq:fading6},
\end{eqnarray}
which holds as $\delta_0 L \to 0$ as $L\to \infty$ in \eqref{eq:PA1}.
(In fact, due to~\cite[Lemma 4 and Lemma 6]{Maurer:Information00},
for any $\epsilon'>0$, we can take
$\delta_0 L \leq \frac{\epsilon'}{L}$ as
$L$ gets large.)
Therefore, for this choice of $L,M,N$,
we obtain the desired result from
\eqref{eq:fading1}, \eqref{eq:fading2}, \eqref{eq:fading4},
due to \eqref{eq:fading5} and \eqref{eq:fading6}:
\begin{eqnarray}
H(W^*) &\geq& L M N \left(E\left[ C(W_m) \right]-\epsilon^*\right) \\
\Prob\{W^*\neq \hat{W}^*\} &\leq& \epsilon \\
I(K;Y_e^*,G) &\leq& \epsilon
\end{eqnarray}
In addition, for this choice of $L,M,N$, we bound
$H(K)\geq r-\epsilon$ due to \eqref{eq:fading3},
which shows that the key is approximately uniform.
\end{proof}

Few remarks are now in order.

1) Existing code designs in the literature and the previous section
of this work assume that Eve's channel is known
at Alice and Bob. In the above scheme,
Alice and Bob only need the statistical knowledge of eavesdropper
CSI. Also, the main channel is not necessarily stronger than the
eavesdropper channel, which is not the case for degraded wiretap settings.

2) The above scheme can be used for the wiretap channel of
Section~\ref{sec:WiretapChannel} by setting $M=0$
to achieve strong secrecy (assuring arbitrarily small information leakage)
instead of the weak notion (making the leakage rate small).
See also~\cite{Maurer:Information00}.

3) The results can be extended to arbitrary binary-input channels
along the same lines, using the result of Section~\ref{sec:WiretapChannel}.
In such a setting, the above theorem would be reformulated
with $n=LMN (E[I(W_m)]-\epsilon^*)$ and
$r=LMN (E[[I(W_m)-I(W_e)]^+]-\epsilon^*)$. However,
the code construction complexity of such channels may not
scale as good as that of the erasure channels~\cite{Arikan:Channel09}.

%%%%%%%%%%%%%%%%%%%%%%%%%%%%%%%%%%%%%%%%%%%%%%%%%%%%%%%%%%%%%%%%%%%%%%%%%%%%%%
%%%%%%%%%%%%%%%%%%%%%%%%%%%%%%%%%%%%%%%%%%%%%%%%%%%%%%%%%%%%%%%%%%%%%%%%%%%%%%

\section{Discussion}
\label{sec:Conclusion}

In this work, we considered polar coding for binary-input DMCs
with a degraded eavesdropper. We showed that polar coding
can be utilized to achieve non-trivial secrecy rates these
set of channels. The results might be extended to arbitrary
discrete memoryless channels using the techniques given
in~\cite{Sasoglu:Polarization}.
The second focus of this work was the secret key agreement
over fading channels, where we showed that Alice and Bob can
create advantage over Eve by using the polar coding scheme at
each fading block, which is then exploited with privacy amplification
techniques to generate keys. This result is interesting in the sense
that part of the key agreement protocol is established information
theoretically over fading channels by only requiring statistical
knowledge of eavesdropper CSI at the users.

%%%%%%%%%%%%%%%%%%%%%%%%%%%%%%%%%%%%%%%%%%%%%%%%%%%%%%%%%%%%%%%%%%%%%%%%%%%%%%
%%%%%%%%%%%%%%%%%%%%%%%%%%%%%%%%%%%%%%%%%%%%%%%%%%%%%%%%%%%%%%%%%%%%%%%%%%%%%%

%%%%%%%%%%%%%%%%%%%%%%%%%%%%%%%%%%%%%%%%%%%%%%%%%%%%%%%%%%%%%%%%%%%%%%%%%%%%%%
%%%%%%%%%%%%%%%%%%%%%%%%%%%%%%%%%%%%%%%%%%%%%%%%%%%%%%%%%%%%%%%%%%%%%%%%%%%%%%

\bibliographystyle{IEEEtran}
%\bibliography{IEEEabrv,D:/OxO/Latex/_References/myref}

\begin{thebibliography}{10}
\providecommand{\url}[1]{#1}
\csname url@samestyle\endcsname
\providecommand{\newblock}{\relax}
\providecommand{\bibinfo}[2]{#2}
\providecommand{\BIBentrySTDinterwordspacing}{\spaceskip=0pt\relax}
\providecommand{\BIBentryALTinterwordstretchfactor}{4}
\providecommand{\BIBentryALTinterwordspacing}{\spaceskip=\fontdimen2\font plus
\BIBentryALTinterwordstretchfactor\fontdimen3\font minus
  \fontdimen4\font\relax}
\providecommand{\BIBforeignlanguage}[2]{{%
\expandafter\ifx\csname l@#1\endcsname\relax
\typeout{** WARNING: IEEEtran.bst: No hyphenation pattern has been}%
\typeout{** loaded for the language `#1'. Using the pattern for}%
\typeout{** the default language instead.}%
\else
\language=\csname l@#1\endcsname
\fi
#2}}
\providecommand{\BIBdecl}{\relax}
\BIBdecl

\bibitem{Shannon:Communication49}
C.~E. Shannon, ``Communication theory of secrecy systems,'' \emph{The Bell
  System Technical Journal}, vol.~28, pp. 656--715, 1949.

\bibitem{Wyner:The75}
A.~Wyner, ``The wire-tap channel,'' \emph{The Bell System Technical Journal},
  vol.~54, no.~8, pp. 1355--1387, Oct. 1975.

\bibitem{Csiszar:Broadcast78}
I.~Csisz{\'a}r and J.~{K\"{o}rner}, ``Broadcast channels with confidential
  messages,'' \emph{{IEEE} Trans. Inf. Theory}, vol.~24, no.~3, pp. 339--348,
  May 1978.

\bibitem{Leung-Yan-Cheong:The78}
S.~Leung-Yan-Cheong and M.~Hellman, ``The gaussian wire-tap channel,''
  \emph{{IEEE} Trans. Inf. Theory}, vol.~24, no.~4, pp. 451--456, Jul. 1978.

\bibitem{Thangaraj:Application07}
A.~Thangaraj, S.~Dihidar, A.~R. Calderbank, S.~W. McLaughlin, and J.-M.
  Merolla, ``Applications of {LDPC} codes to the wiretap channel,''
  \emph{{IEEE} Trans. Inf. Theory}, vol.~53, no.~8, pp. 2933--2945, Aug. 2007.

\bibitem{Liu:Secure07}
R.~Liu, Y.~Liang, H.~V. Poor, and P.~Spasojevic, ``Secure nested codes for type
  {II} wiretap channels,'' in \emph{Proc. {IEEE} Information Theory Workshop
  (ITW'07)}, Sep. 2007.

\bibitem{Liang:Information08}
Y.~Liang, H.~V. Poor, and S.~{Shamai (Shitz)}, ``Information theoretic
  security,'' \emph{Foundations and Trends in Communications and Information
  Theory}, vol.~5, no. 4-5, pp. 355--580, 2008.

\bibitem{Arikan:Channel09}
E.~Ar{\i}kan, ``Channel polarization: A method for constructing
  capacity-achieving codes for symmetric binary-input memoryless channels,''
  \emph{{IEEE} Trans. Inf. Theory}, vol.~55, no.~7, pp. 3051--3073, Jul. 2009.

\bibitem{Hessam:Achieving}
\BIBentryALTinterwordspacing
H.~Mahdavifar and A.~Vardy, ``Achieving the secrecy capacity of wiretap
  channels using polar codes,'' 2010, submitted for publication. [Online].
  Available: \url{http://arxiv.org/abs/1001.0210}
\BIBentrySTDinterwordspacing

\bibitem{Arikan:On09}
E.~Ar{\i}kan and E.~Telatar, ``On the rate of channel polarization,'' in
  \emph{Proc. 2010 IEEE International Symposium on Information Theory}, Seoul,
  Korea, Jun. 2009.

\bibitem{Korada:Polar09}
S.~B. Korada, ``Polar codes for channel and source coding,'' Ph.D.
  dissertation, {Lausanne, Switzerland}, 2009.

\bibitem{Gopala:On08}
P.~Gopala, L.~Lai, and H.~{El Gamal}, ``On the secrecy capacity of fading
  channels,'' \emph{{IEEE} Trans. Inf. Theory}, vol.~54, no.~10, pp.
  4687--4698, Oct. 2008.

\bibitem{Bennett:Generalized95}
C.~H. Bennett, G.~Brassard, C.~Crepeau, and U.~M. Maurer, ``Generalized privacy
  amplification,'' \emph{{IEEE} Trans. Inf. Theory}, vol.~41, no.~6, pp.
  1915--1923, Nov. 1995.

\bibitem{Bloch:WirelessI}
M.~Bloch, J.~Barros, M.~Rodrigues, and S.~W.McLaughlin, ``Wireless
  information-theoretic security,'' \emph{{IEEE} Trans. Inf. Theory}, vol.~54,
  no.~6, pp. 2515--2534, Jun. 2008.

\bibitem{Carter:Universal79}
J.~L. Carter and M.~N. Wegman, ``Universal classes of hash functions,''
  \emph{J. Comput. Syst. Sci.}, vol.~18, pp. 143--154, 1979.

\bibitem{Maurer:Information00}
U.~Maurer and S.~Wolf, ``Information-theoretic key agreement: From weak to
  strong secrecy for free,'' in \emph{Advances in Cryptology - {EUROCRYPT}
  2000, Lecture Notes in Computer Science 1807}, 2000, pp. 351--368.

\bibitem{Sasoglu:Polarization}
\BIBentryALTinterwordspacing
E.~Sasoglu, E.~Ar{\i}kan, and E.~Telatar, ``Polarization for arbitrary discrete
  memoryless channels,'' 2009, submitted for publication. [Online]. Available:
  \url{http://arxiv.org/abs/0908.0302}
\BIBentrySTDinterwordspacing

\end{thebibliography}


%%%%%%%%%%%%%%%%%%%%%%%%%%%%%%%%%%%%%%%%%%%%%%%%%%%%%%%%%%%%%%%%%%%%%%%%%%%%%%
%%%%%%%%%%%%%%%%%%%%%%%%%%%%%%%%%%%%%%%%%%%%%%%%%%%%%%%%%%%%%%%%%%%%%%%%%%%%%%

\end{document}